\documentclass[10pt,journal,compsoc]{IEEEtran}
\usepackage{setspace}
\usepackage{amsfonts}
\usepackage{amssymb}
\usepackage{stmaryrd}
\usepackage{bbding}
\usepackage{algorithm}
\usepackage{algorithmic}
\usepackage{subfigure}
\usepackage{graphicx}
\usepackage{color}
\usepackage{url}
\graphicspath{{figures/}}

\usepackage{amsthm}

\ifCLASSOPTIONcompsoc
  \usepackage[nocompress]{cite}
\else
  \usepackage{cite}
\fi
\newtheorem{mydef}{Definition}
\newtheorem{mytheorem}{Theorem}
\newtheorem{mylemma}{Lemma}

\ifCLASSINFOpdf
\else
\fi
\usepackage{amsmath}
\begin{document}
%
\title{Differentially Private Combinatorial Cloud Auction}
%
%
%
%

\author{Tianjiao Ni,
        Zhili Chen*, Lin Chen, Hong Zhong, Shun Zhang, Yan Xu
\IEEEcompsocitemizethanks{\IEEEcompsocthanksitem Tianjiao Ni,
        Zhili Chen, Hong Zhong, Shun Zhang, Yan Xu are with the School of Computer Science and Technology, Anhui University, Hefei, China, 230601. Zhili Chen is the corresponding author.\protect\\ 
        E-mail:tjni94@163.com; zlchen@ahu.edu.cn;  zhongh@mail.ustc.edu.cn; shzhang27@163.com; xuyan@ ahu.edu.cn.
\IEEEcompsocthanksitem Lin Chen is with Lab. Recherche Informatique (LRI-CNRS UMR 8623), Univ. Paris-Sud, 91405 Orsay, France.\protect\\
E-mail:chen@lri.fr.}}

%
%

\markboth{IEEE TRANSACTIONS ON DEPENDABLE AND SECURE COMPUTING.,~Vol.~0, No.~0, January~2020}%
{Shell \MakeLowercase{\textit{et al.}}: Bare Demo of IEEEtran.cls for Computer Society Journals}
%



\IEEEtitleabstractindextext{%
\begin{abstract}
Cloud service providers typically provide different types of virtual machines (VMs) to cloud users with various requirements. Thanks to its effectiveness and fairness, \textit{auction} has been widely applied in this heterogeneous resource allocation. Recently, several strategy-proof combinatorial cloud auction mechanisms have been proposed. However, they fail to protect the bid privacy of users from being inferred from the auction results. In this paper, we design a \textit{differentially private} combinatorial cloud auction mechanism (DPCA) to address this privacy issue. Technically, we employ the exponential mechanism to compute a clearing unit price vector with a probability proportional to the corresponding revenue. We further improve the mechanism to reduce the running time while maintaining high revenues, by computing a single clearing unit price, or a subgroup of clearing unit prices at a time, resulting in the improved mechanisms DPCA-S and its generalized version DPCA-M, respectively. We theoretically prove that our mechanisms can guarantee differential privacy, approximate truthfulness and high revenue. Extensive experimental results demonstrate that DPCA can generate near-optimal revenues at the price of relatively high time complexity, while the improved mechanisms achieve a tunable trade-off between auction revenue and running time.
\end{abstract}

}

\maketitle

\IEEEdisplaynontitleabstractindextext

%
\IEEEpeerreviewmaketitle

\IEEEraisesectionheading{\section{Introduction}\label{sec:introduction}}

%
%
%
%
\IEEEPARstart{C}{loud} computing provides a platform for a large number of users to access computing resources such as CPU and memory. Today, cloud service providers (e.g., Amazon EC2 and Microsoft Azure) typically use virtualization techniques to configure their resources as different types of Virtual Machines (VM) and then sell them to cloud users with fixed-price methods. Regrettably, fixed-price allocation mechanisms cannot reflect the dynamic supply-demand relationship of the market and may lead to economic inefficiency. 
To mitigate this problem, auction-based pricing polices have emerged in the cloud market. For example, the Amazon EC2 Spot Instance \cite{Amazon.org} has employed an auction mechanism to allocate idle VM instances to users. Auction has been proven to be an effective market-based cloud service transaction mechanism not only allowing users to get the resources they need at appropriate prices, but also enabling cloud providers to leverage more resources to improve their profits.

In an auction mechanism, truthfulness (a.k.a. strategy-proofness) is one of the most crucial economic properties. A truthful auction incentivizes bidders to bid their true valuations. Recently, a number of truthful cloud auction mechanisms~\cite{Wang2012When,Zaman2010Combinatorial,Shi2014RSMOA} have been developed. Unfortunately, maximizing utility therein comes at the expense of disclosing bid information, which reflects the bidders' preferences and demands for VM resources. Adversaries may utilize these bids to infer bidders' computing requirements, the workload patterns and etc. These information may be vital commercial secrets. Therefore, it is important to protect bidders' bid privacy.

To prevent information leakage, only a handful of researches \cite{Chen2016On} has focused on privacy protection in cloud auction, which protects privacy from a dishonest auctioneer in the process of auction computations. But it did not consider that an attacker can infer the bid information based on the published outcome. Usually, the auctioneer is considered to be trusted in a sealed auction, and he will publish the auction results including the winners and their payments. In most of the cloud auctions, VM resources are only available to bidders for a certain period of time, bidders may participate in cloud auction frequently. However, once bidders submit the true valuation to the reliable  auctioneer, an attacker (who may be a bidder or anyone seeing these outcomes except the auctioneer) may infer some private information of bidders from outcomes of two or more auctions.

To address this privacy issue, the notion of differential privacy \cite{Dwork2006Differential}, which can preserve privacy for published results with strong theoretical guarantee, has been applied as a compelling privacy model. While there exist a number of differentially private auction-based mechanisms \cite{mcsherry2007mechanism,zhu2014differentially,Jin2016Enabling,Jian2017BidGuard} to limit privacy leakage, cloud auction with differential privacy guarantee \cite{xu2017pads} remains largely unaddressed, especially for the combinatorial auctions that are the most suitable for solving VM pricing and allocation problems. Unlike traditional auctions,
design a differentially private auction mechanism in heterogeneous environment is a non-trivial task. On the one hand, different users request a set of VM instances of different types. A reasonable pricing and allocation strategies should depend on the VM instances they request. Meanwhile, it should motivate bidders to report their true valuation while guaranteing the revenue of the provider.
On the other hand, reporting the true valuations may disclose bidders' privacy information ( bids and the number of VM instances ). The auction mechanism should be designed to achieve both truthfulness and privacy protection.


To cope with the aforementioned challenges, we develop a Differentially Private Combinatorial Cloud Auction (DPCA) mechanism in this paper, where a cloud service provider provides various types of VMs to heterogeneous users. To keep users' bids private and achieve high revenue of the provider, we leverage the exponential mechanism to select the final clearing unit price vector, exponentially proportional to the corresponding revenue. Technically, the mechanism is carefully designed to meet the instance constraint and achieve approximate truthfulness. Armed with the baseline mechanism, we further reduce the running time while maintaining the high revenue, by computing a single clearing unit price, or a subgroup of clearing unit prices at a time, resulting in the improved mechanisms DPCA-S and its generalized version DPCA-M, respectively.

Our main contributions can be summarized as follows.
\begin{itemize}
   \item To the best of our knowledge, we are the first to apply differential privacy in combinatorial cloud auction. Our proposition DPCA can not only prevent the disclosure of bidders' bids, but also ensure approximate truthfulness and high revenue.
    \item To reduce computational complexity, we propose DPCA-S which uses an exponential mechanism to select single unit price of one type of VM each time. Furthermore, we design DPCA-M to select the combination of multiple unit prices at a time to improve revenue.
    \item We fully implement the proposed mechanisms, and conduct extensive experiments to evaluate their performances. The experimental results demonstrate that DPCA can generate near-optimal revenue at the price of relatively high time complexity, while the improved mechanisms achieve a tunable trade-off between auction revenue and running time.
\end{itemize}

The remainder of this paper is organized as follows.
Section~\ref{sec:relatedwork} briefly reviews the related work. In Section~\ref{sec:preliminaries}, we introduce auction model and some important concepts. And we present the detailed design of DPCA and prove the related properties in Section \ref{sec:dpca}. And we further propose the improved mechanisms in Section \ref{sec:improve}. In Section \ref{sec:experiment}, we implement our mechanisms, and evaluate their performances. Finally, the paper is concluded in Section
\ref{sec:conclusion}.

\section{Related Work}\label{sec:relatedwork}
As an effective transaction method, auction has been extensively used in the field of cloud computing. Wang $et$ $al.$ \cite{Wang2012When} proposed a computationally efficient and truthful cloud auction mechanism. And Zaman $et$ $al.$ \cite{Zaman2010Combinatorial} designed a combinatorial auction mechanism for the case of static provisioning of VM instances, which can capture the competition of users. To satisfy user's demand and generate higher revenue or social welfare, \cite{Zaman2014A, Mashayekhy2015A,Wang2013Revenue} proposed the truthful mechanisms for dynamic VM provisioning and allocation. Zhang $et$ $al.$ \cite{Zhang2014Dynamic} designed a randomized combinatorial auction by leveraging a pair of primal and dual linear programs (LPs). Du $et$ $al.$\cite{du2019learning} studied a deep reinforcement learning approach in the cloud resource allocation and pricing to maximize the provider's revenue. In paper \cite{Wei2013Dominant}, the authors presented a generalized dominant resource fairness mechanism in heterogeneous environments. Recently, online auction mechanisms \cite{Shi2014RSMOA, Hong2013A, Mashayekhy2016An} have attracted a great deal of attention. Zhang $et$ $al.$ \cite{Hong2013A} utilized the pricing-curve method and proposed a truthful online auction with one type of VM. Later, papers \cite{Shi2014RSMOA, Mashayekhy2016An} took into account heterogeneous VM instances. And in the Blockchain networks, the paper \cite{jiao2019auction} proposed an auction-based market model to realize the efficient allocation of computing resources for the transactions between the cloud/fog computing service provider and miners. However, none of these researches addressed the privacy concerns.

In recent years, people have paid more attention to the issue of privacy. To the best of our knowledge, only a few auction-based cloud resource allocation mechanisms took privacy into consideration. Chen $et$ $al.$ \cite{Chen2016On} designed a privacy-preserving cloud auction with cryptographical techniques. Later, for the two-sided cloud market, Cheng $et$ $al.$ \cite{cheng2019towards} proposed a efficient privacy-preserving double auction mechanism. It did not disclose any information about bid other than the auction results. To prevent the attacker from inferring sensitive information based on published outcome, the concept of differential privacy is proposed. McSherry $et$ $al.$ \cite{mcsherry2007mechanism} proposed the first differentially private auction mechanism. Then, differential privacy has been applied to specifical auction-based scenarios including spectrum auctions \cite{zhu2014differentially,zhu2015differentially,8676063}, spectrum sensing \cite{jin2018privacy} and mobile crowd sensing \cite{Jin2016Enabling, Jian2017BidGuard,jin2018dpda,gao2019dpdt}. In paper \cite{xu2017pads}, incorporating differential privacy, the authors proposed an auction-based mechanism for resource allocation in cloud in the presence of only one type of VM. As far as we know, none of existing works proposes the cloud auction mechanism for multiply types of VMs allocation with differential privacy guarantee.

\section{Technical Preliminaries}\label{sec:preliminaries}

In this section, we present the auction model for combinatorial virtual machine provisioning and allocation in cloud computing, and review some important concepts to facilitate the comprehension of our schemes.

\subsection{Cloud Auction Model}
We consider a sealed-bid cloud resource allocation auction with an auctioneer (seller) and a group of buyers, which is completed in a time windom.
A Cloud Service Provider (CSP), i.e., the auctioneer who is trustworthy, sets $m$ different types of virtual machine (VM) instances represented by $\mathbb{M}=\{VM_1, VM_2,..., VM_m\}$. These $m$ types of instances are mainly composed of different CPU powers (e.g., dual-core, quad-core), memory sizes (e.g., small, medium, large), and operating systems (e.g., Linux/UNIX, Microsoft Windows), etc. Each type of VM has only a certain number of instances. Let $\mathbf{K}=\{K_1, K_2, ..., K_m\}$ denote the profile of all types of VM instances, where $K_i$ represents the number of $VM_i$ instances offered by the CSP with values in the range $[K_{min},K_{max}]$.

Suppose there are $n$ users as buyers requesting these VM instances. Here $\mathbb{N}=\{1, 2, ..., n\}$ denotes the set of users. Each user $j$ submits his bid profile $B_j=<k_j,b_j>$, where $k_j=\{k_j^1, k_j^2, ..., k_j^m\}$ denotes the profile of numbers of $m$ different types of VM instances requested by user $j$, and $b_j=\{b_j^1, b_j^2, ..., b_j^m\}$ represents the profile of the per-instance bids for $m$ different types of VM instances submitted by user $j$. Here, we assume that the number of $VM_i$ instances that user $j$ requests is in the interval $[0, q_{max}]$. And it should meet the requirement of $q_{max} \ll K_i$ (If the user requests a large number of resources, he can work directly with the provider, without the need for the auction). The bid of user $j$ is based on his true valuation for the same VM instances denoted by $v_j=\{v_j^1, v_j^2, ..., v_j^m\}$, and $v_j^i$ is limited in the range of $[ v_{min}, v_{max}]$. If the user $j$ is not interested in $VM_i$ , then $v_j^i=0$. The total bid of user $j$ for his desired VM bundle is $\overline{B}_j=\sum_{i=1}^m {k_j^i}{b_j^i}$. And the corresponding total valuation is $\overline{V}_j=\sum_{i=1}^m {k_j^i}{v_j^i}$.
The profile of bid of all users is represented by $\mathbf{B}=\{B_1, B_2, ..., B_n\}$.

The outcome of an auction contains the profiles of allocation and payment. Let $\mathbf{x}=\{x_1,x_2...,x_n\}$ denote the allocation profile, where $x_j \in \{0,1\}$ indicates whether the user $j$ gets the requested VM instances. The payments are represented by $\mathbf{P}=\{P_1,P_2,...,P_n\}$ where $P_j$ is the price that user $j$ ultimately pays. Here we consider the users are single-minded. They are only willing to pay for the VM instances they request, and for other VM instances, the payments are zero (i.e., $P_j=0$). 

In an auction, each user $j$ is considered selfish and rational, and his objective is to maximize his own utility $u_j$:
\begin{displaymath}
u_j=\overline{V}_j x_j-P_j
\end{displaymath}
The CSP also wants to maximize his revenue, where the revenue is the sum of payments of all users:
\begin{displaymath}
REV=\sum_{j=1}^n P_j x_j
\end{displaymath}

\subsection{Attack model}

In the truthful cloud auction mechanism where the auctioneer is trusted, it requires each bidder to submit his true valuation. But once the true valuation is submitted, the attacker (anyone other than the auctioneer) can infer the bidder's private information from the public auction outcome. 
In most of the cloud auctions, VM instances are only available for short-term use by bidders, so bidders may participate in auctions frequently. This makes it easier to leak private information.
\begin{figure}[htbp]
\begin{center}
\includegraphics[width=0.3\textwidth]{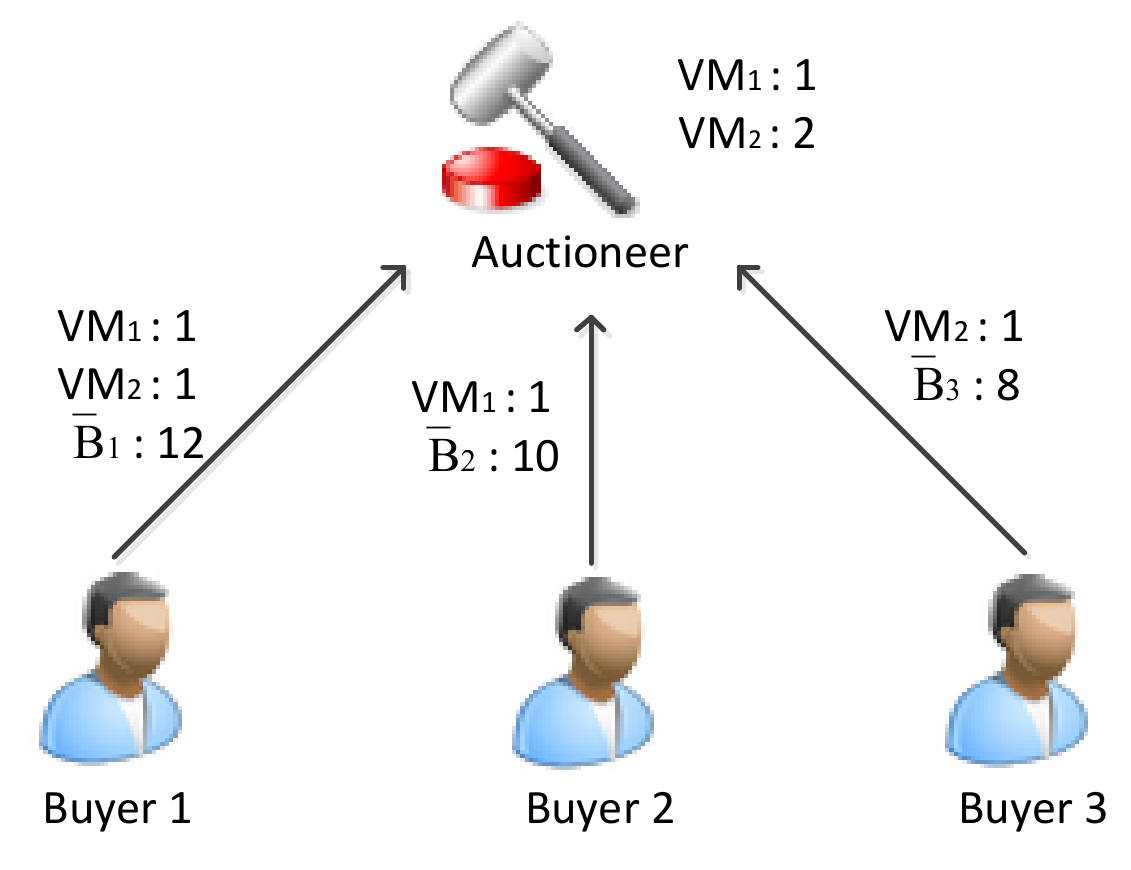}
\caption{A cloud auction attack example }\label{fig:example}
\end{center}
\end{figure}
\begin{figure}
\end{figure}

To better understand this attack model, we present an attack example as shown in Fig.~\ref{fig:example}. Suppose there are two types of VM with one $VM_1$ instance and two $VM_2$ instances. And three buyers compete for these instances, in which $Buyer\ 1$ requests one $VM_1$ instance and one $VM_2$ instance with total bid 12, $Buyer\ 2$ requests one $VM_1$ instance with total bid 10 and $Buyer\ 3$ requests one $VM_2$ instance with total bid 8. According to the truthful cloud auction \cite{Wang2012When}, the winners are $Buyer\ 2$ and $Buyer\ 3$. Here we simplify the ranking metric \cite{Wang2012When} to \textit{average bid per each instance}. And the payments of $Buyer\ 2,Buyer\ 3$ are 6 and 0, respectively. Assume that $Buyer\ 2$ is an attacker, and he changes his total bid from 10 to 5 in the second auction, while other buyers' bids remain unchanged. Then the auction result becomes that winners are $Buyer\ 1$ and $Buyer\ 3$ and their payments are 10 and 0, respectively.
According to the fact that $Buyer\ 3$ pays 0 in the second auction, it can be known that he has no \textit{critical bidder} \cite{Wang2012When}. And it is easy for $Buyer\ 2$ to deduce that $Buyer\ 3$ only requests $VM_2$ and his \textit{average bid per each instance} is more than 5, since if $Buyer\ 3$ requests $VM_1$, the loser $Buyer\ 2$ requesting one $VM_1$ instance will be his \textit{critical bidder}. It is also known that the \textit{average bid per each instance} of $Buyer\ 1$ is between 5 and 10 by the fact that if  $Buyer\ 2$ bids 10, he will win, and if he bids 5, he will be replaced by $Buyer\ 1$. So the privacy of $Buyer\ 1$ and $Buyer\ 3$ is compromised.

\subsection{Mechanism Design Concepts}

Now we introduce some related solution concepts from mechanism design and differential privacy.
First, we review the definition of the dominant strategy in an auction mechanism.
\begin{mydef}[Dominant Strategy \cite{Nisan2007Algorithmic}]
Strategy $s_i$ is a player $i$'s dominant strategy in a game, if for any strategy ${s'_i}$ $\neq$ $s_i$ and other players' strategy profile $s_{-i}$,
\begin{displaymath}
u_i(s_i, s_{-i}) \geq u_i({s'_i}, s_{-i})
\end{displaymath}
\end{mydef}
In an auction, truthfulness is associated with a dominant strategy meaning that users reveal their truthful bids. But in some cases, it is too restrictive to satisfy exact truthfulness. So we turn to consider a weaker but more practical version of truthfulness, that is, approximate truthfulness or $\gamma$-truthfulness\cite{Gupta2010Differentially}.
\begin{mydef}[$\gamma$-truthfulness]
Let $\overline{V}_j$ denote the total bid when user $j$ bids truthfully. A auction mechanism is $\gamma$-truthfulness in expectation, if and only if for any total bid $\overline{B}_j \neq \overline{V}_j$ and other users' total bid profile $\mathbf{\overline{V}_{-j}}$, it satisfies
\begin{displaymath}
E[u_j(\overline{V}_j, \mathbf{\overline{V}_{-j}})] \geq E[u_j(\overline{B}_j, \mathbf{\overline{V}_{-j}})]-{\gamma}
\end{displaymath}
where $\gamma$ is a small constant.
\end{mydef}

Differential privacy \cite{Dwork2006Differential,dwork2008differential} ensures that for any two databases differing in one record, the probability of outputting the same result is essentially identical. Thus it is difficult for an adversary to infer users' private information from the auction results.
\begin{mydef}[Differential Privacy]
Let $\mathbf{B}$ denote the bid profile of all users for required VM instances. A randomized cloud auction mechanism $\mathcal{M}$ is $\epsilon$-differentially private if for any two two data profiles $\mathbf{B}$ and $\mathbf{B'}$ with only one different bid, and $S$ $\subseteq$ Range($\mathcal{M}$), it satisfies
\begin{displaymath}
Pr[\mathcal{M}(\mathbf{B})\in S]\leq exp(\epsilon)\times Pr[\mathcal{M}(\mathbf{B'})\in S]
\end{displaymath}
\end{mydef}
where $\epsilon$ is a small positive constant called privacy budget meaning the level of privacy protection.

Exponential mechanism \cite{mcsherry2007mechanism} is an effective tool to realize differential privacy. It determines a utility score for any possible output, and the higher score is, the more likely the candidate output will be chosen. Specifically, it is defined as follows:
\begin{mydef}[Exponential Mechanism]
Given a range $\mathcal{P}$ and a utility function $Q(\mathbf{B},p)$ which maps a pair of the input profile $\mathbf{B}$ and an output $p$ in the range $\mathcal{P}$ to a real-valued score. The exponential mechanism $\mathcal{M}(\mathbf{B}, Q, \mathcal{P})$ selects and outputs $p \in \mathcal{P}$ with probability
\begin{displaymath}
Pr[\mathcal{M}(\mathbf{B}, Q, \mathcal{P})= p]\propto exp(\frac{\epsilon Q(\mathbf{B}, p)}{2\Delta Q})
\end{displaymath}
where $\Delta Q$ is the sensitivity of function $Q$, that is, for any $p \in \mathcal{P}$ and for any two profiles $\mathbf{B}$ and $\mathbf{B'}$ differing in a single element,  the largest change in $Q(\mathbf{B}, p)$ and $Q(\mathbf{B'}, p)$, and $\epsilon$ is the privacy budget.
\end{mydef}

\begin{mytheorem}[\cite{dwork2014algorithmic}]\label{the:dp}
The exponential mechanism guarantees $\epsilon$-differential privacy.
\end{mytheorem}

\begin{mylemma}[Composition \cite{dwork2014algorithmic}]\label{the:composition}
Given the randomized algorithms $\mathcal{M}_1$, $\mathcal{M}_2$,..., $\mathcal{M}_k$ that satisfies $\epsilon_1$-differential privacy, $\epsilon_2$-differential privacy,..., $\epsilon_k$-differential privacy, respectively. Then $\mathcal{M}(D)=(\mathcal{M}_1(D), \mathcal{M}_2(D),..., \mathcal{M}_k(D))$ satisfies $(\sum_{i=1}^k \epsilon_i)$-differential privacy.
\end{mylemma}

\section{Differentially Private Combinatorial Cloud Auction (DPCA)}\label{sec:dpca}

In this section, we elaborate our differentially private mechanism for combinatorial cloud auctions, DPCA, and provide theoretical analysis of its related properties.

\subsection{Design Rationales}

We integrate the concept of differential privacy with the problem of heterogeneous cloud resource allocation, and design a combinatorial cloud auction mechanism which preserves differential privacy, achieves truthfulness, and realizes good revenues. Two main challenges and their corresponding design rationales are as follows.


The first challenge is how to design a pricing strategy to ensure reasonable pricing, which is good for achieving the property of truthfulness and yielding high auction revenues. Our observation is that, for heterogeneous cloud auctions, the total payment of each winner is dependent on both unit prices and numbers of VMs, and thus the pricing for different users should have similar dependence. Based on this observation, we apply a two-level pricing. The first level pricing is a vector of all unit prices, providing a common pricing; The second level is total prices for users computed based on both unit prices and numbers of VMs, providing user-related pricing. Then the payments of users is proportional to the numbers of VMs they requested, which is easy for users to accept.

Since there is a risk of privacy disclosure in submitting the true valuations, how to design the VM allocation procedure to achieve both differential privacy and truthfulness is the second challenge. To address this challenge, we first apply the exponential mechanism to select a unit price vector, achieving the differential privacy. Then, to achieve truthfulness, we design a random VM allocation algorithm by first selecting a winner candidate set through the unit price vector, and then allocating VMs to winner candidates in a random order with instance constraints. Here, instance constraints refer to the conditions that all numbers of VMs allocated should be no more than those provided.

\subsection{Detailed Design}
In the design, the domain of each unit price $\rho_i$ of $VM_i$ ($0 \le i \le m$) is defined as $\Pi=[v_{min}..v_{max}]$, where each unit price $\rho_i$ can take all the different possible valuation/bid values in $\Pi$. As a result, the unit clearing prices can be represented by a vector $\mathbf{\rho}=(\rho_1,\rho_2,...,\rho_m) \in {\Pi}^m$. The procedure of DPCA can be described in the following two phases.

\textbf{(1) Price Vector Selection.} In this phase, DPCA applies an instance of the exponential mechanism to select the clearing price vector, where revenues are used as the utility function. This phase can be further divided into three steps as follows.

\emph{(a) Winner Candidate Selection.}This step selects winner candidates through price comparison. Specifically, given the clearing price vector $\rho \in {\Pi}^m$, the clearing price of user $j$ can be computed as
\begin{equation}
P_j=\sum_{i=1}^m k_j^i \rho_i
\end{equation}
Then user $j$ is selected as a candidate if
\begin{equation}\label{eq:candidate3}
{\overline{B}_j}\geq {P_j}
\end{equation}


Through this step, the users, whose total bids are not less than their corresponding total payments at price $\rho$, are selected as winner candidates. We denote the set of candidates by $\mathcal{W}^c$ and assume that $h = |\mathcal{W}^c|$.

\emph{(b) Random VM Allocation.} In this step, the CSP sorts these $h$ candidates randomly, such that whether a candidate is a winner is totally independent of its bid. This randomness is important for the achievement of truthfulness as we will discuss later in the theoretical proofs. However, the randomness is actually predetermined by the auctioneer before the auction. We adopt the idea of the matching random string \cite{8676063} to ensure that the revenues corresponding to all clearing prices are unique, and hence revenues can be used as a utility function for the exponential mechanism. Here we use $r$ to represent the random string. Specially, $r$ is a sufficiently long bit string that can determine an order of any candidate set for all possible price vectors. And each part of $r$ indicates the random sort of candidates of a certain price vector $\rho$.
Then given a price vector $\rho$ as the clearing price vector, and an allocation order represented by $r$, our mechanism chooses the first user that satisfies the instance constraint (e.g., the VM instances requested by the user does not exceed the amount of VM instances provided by the provider) as the winner while updating the instance constraint, and then chooses the next winner as above. The iteration continues until all users in $\mathcal{W}^c$ have been examined, and the winner set $\mathcal{W}$ is obtained. Then the corresponding revenue of the CSP can be calculated as follows.
\begin{equation}\label{eq:revenue}
REV(\mathbf{B}, \mathbf{K}, r,\rho)=\sum_{j \in \mathcal{W}} \sum_{i=1}^m \rho_i k_j^i
\end{equation}

\emph{(c) Probabilistic Price Selection.} Steps (a) and (b) are carried out given a price vector $\rho$. For the achievement of privacy protection and better revenue, we introduce the exponential mechanism to select the price $\rho$ with revenues of the CSP as its utility function. Thus, this step needs to repeatedly perform the previous two steps to get the revenues corresponding to all possible price vectors $\rho$, given a random string $r$. Then the probability distribution of clearing price vectors can be computed as follows.
\begin{equation}\label{eq:probability}
Pr(\mathcal{M}(\mathbf{B}, \mathbf{K}, r) = \rho)=\frac{{\exp}(\frac{\epsilon REV(\mathbf{B}, \mathbf{K}, r, \rho)}{2\Delta})}{\sum_{\rho' \in {\Pi}^m}{\exp}(\frac{\epsilon REV(\mathbf{B}, \mathbf{K}, r, \rho')}{2\Delta})}
\end{equation}
where $\epsilon$ is the privacy budget, and $\Delta=m \cdot q_{max} v_{max}$, i.e., the possible maximum change in revenue that alters one's bid.

Based on the probability distribution \eqref{eq:probability}, DPCA randomly chooses a final clearing unit price vector $\rho$. Algorithm~\ref{protocol:DPCA} describes the procedure in details.

\textbf{(2) Winner Computation.}
Given a random string r, once the final clearing price vector $\rho$ is chosen, the winners will be determined accordingly. Then each winner $j$ will pay the CSP $P_j=\sum_{i=1}^m \rho_i k_j^i$. And the final revenue of the CSP is $REV(\mathbf{B}, \mathbf{K}, r, \rho)$.

\renewcommand{\algorithmicrequire}{\textbf{Input:}}  
\renewcommand{\algorithmicensure}{\textbf{Output:}}  

\begin{algorithm}[htb]
\caption{Price Vector Selection} \label{protocol:DPCA}
\begin{algorithmic}[1]

\REQUIRE $\mathbf{B}$, $\mathbf{K}$, $\mathbb{N}$, $r$ and privacy budget $\epsilon$.

\ENSURE Final clearing price $\rho$

\STATE Initialize $x_1=...=x_n=0$, $\mathcal{W} \leftarrow \varnothing$, $\mathcal{W}^c \leftarrow \varnothing$
\STATE $\rho=(\rho_1,\rho_2,...,\rho_m) \in \Pi^m$

\FOR{all $\rho \in \Pi^m$}
\FOR{$j=1,...,n$}
\IF{$\sum_{i=1}^m k_j^i b_j^i \geq \sum_{i=1}^m k_j^i \rho_i$}

\STATE $\mathcal{W}^c \leftarrow \mathcal{W}^c \cup j$
\ENDIF
\ENDFOR

\STATE  Rank the $h$ candidates in $\mathcal{W}^c$ according to $r$

\FOR{$j=1,...,h$}
\IF{$\forall i \in \{1,...,m\}, \sum_{t=1}^{j-1}k_t^i x_t+k_j^i\leq K_i$}

\STATE $x_j=1, \mathcal{W} \leftarrow \mathcal{W} \cup j$

\ENDIF
\ENDFOR

\STATE Calculate the revenue of the CSP :

$REV(\mathbf{B}, \mathbf{k}, r, \rho)\leftarrow \sum_{j \in \mathcal{W}} \sum_{i=1}^m \rho_i k_j^i x_j$

\ENDFOR

\STATE Randomly select a final clearing price $\rho$ from $\Pi^m$ according to the following distribution:

$Pr(\rho)=\frac{{\exp}(\frac{\epsilon REV(\mathbf{B}, \mathbf{k}, r, \rho)}{2\Delta})}{\sum_{\rho' \in {\Pi}^m}{\exp}(\frac{\epsilon REV(\mathbf{B}, \mathbf{k}, r, \rho')}{2\Delta})}$


\RETURN $\rho$

\end{algorithmic}
\end{algorithm}

\subsection{Analysis}
In this subsection, we focus on analyzing and proving desirable properties of DPCA, namely its privacy (Theorem~\ref{the:dp}), truthfulness (Theorem~\ref{the:gamma-truthfulness}), and its high revenue (Theorem~\ref{the:OPT}). In the following analysis, for the changes in the bid profile $B_j$ of one user, we take into account the variations of both per-instance bids $b_j^i$ and numbers of VMs $k_j^i$ requested. The inputs of our auction mechanism include the profile of all types of instances and the randomness $r$. But these input values are provided by the CSP and do not change, so we will not write them out explicitly in expressions in the following proof unless otherwise stated.

\begin{mytheorem}\label{the:dp}
DPCA achieves $\epsilon$-differential privacy.
\end{mytheorem}

\begin{proof}
We denote $\mathbf{B}$ as the profile of bids for all users, and denote any neighboring profile by $\mathbf{B'}$, which is only differing in the bid profile $B_j$ of one user from $\mathbf{B}$. Let $\mathcal{M}$ denote the mechanism we designed to randomly select the clearing price vector. For any $\rho \in {\Pi}^m$ and any pair of profiles $\mathbf{B}$ and $\mathbf{B'}$, the probability ratio of the clearing price vector selected by the mechanism $\mathcal{M}$ is

\begin{displaymath}
\begin{aligned}
&\frac{Pr(\mathcal{M}(\mathbf{B})=\rho)}{Pr(\mathcal{M}(\mathbf{B'})=\rho)}\\
=&\frac{\frac{{\exp}(\frac{\epsilon REV(\mathbf{B},
\rho)}{2\Delta})}{\sum_{\rho'\in {{\Pi}^m}}{\exp}(\frac{\epsilon REV(\mathbf{B}, \rho')}{2\Delta})}}{\frac{{\exp}(\frac{\epsilon REV(\mathbf{B'}, \rho)}{2\Delta})}{\sum_{\rho'\in {{\Pi}^m}}{\exp}(\frac{\epsilon REV(\mathbf{B'}, \rho')}{2\Delta})}}\\
=&(\frac{{\exp}(\frac{\epsilon REV(\mathbf{B}, \rho)}{2\Delta})}{{\exp}(\frac{\epsilon REV(\mathbf{B'}, \rho)}{2\Delta})})(\frac{{\sum_{\rho'\in {{\Pi}^m}}{\exp}(\frac{\epsilon REV(\mathbf{B'}, \rho')}{2\Delta})}}{{\sum_{\rho'\in {{\Pi}^m}}{\exp}(\frac{\epsilon REV(\mathbf{B}, \rho')}{2\Delta})}})\\
\leq& \exp(\frac{\epsilon}{2})(\frac{{\sum_{\rho'\in {{\Pi}^m}}\exp(\frac{\epsilon}{2}){\exp}(\frac{\epsilon REV(\mathbf{B}, \rho')}{2\Delta})}}{{\sum_{\rho'\in {{\Pi}^m}}{\exp}(\frac{\epsilon REV(\mathbf{B}, \rho')}{2\Delta})}})\\
\leq & \exp(\frac{\epsilon}{2})\exp(\frac{\epsilon}{2})(\frac{{\sum_{\rho'\in {{\Pi}^m}}{\exp}(\frac{\epsilon REV(\mathbf{B}, \rho')}{2\Delta})}}{{\sum_{\rho'\in {{\Pi}^m}}{\exp}(\frac{\epsilon REV(\mathbf{B}, \rho')}{2\Delta})}})\\
=& \exp(\epsilon)
\end{aligned}
\end{displaymath}

Symmetrically, we have $\exp(-\epsilon) \le \frac{Pr(\mathcal{M}(\mathbf{B})=\rho)}{Pr(\mathcal{M}(\mathbf{B'})=\rho)}$. Thus, DPCA achieves $\epsilon$-differential privacy.
\end{proof}

\begin{mylemma}\label{the:truthfulness}
Given a random string $r$ and a clearing price vector $\rho \in \Pi^m$, any user $j$ cannot improve its utility by bidding untruthfully, i.e., $u_j(V_j, \mathbf{B_{-j}})\geq u_j(B_j, \mathbf{B_{-j}})$.
\end{mylemma}
\begin{proof}
In the following proof, we use $x_j^*$ and $u_j^*$ to represent the allocation and utility of user $j$ when his true valuation is made, i.e., $B_j=V_j$.
If $B_j\neq V_j$, it means that some per-instance bids or some numbers of VMs requested are different from true valuation's, then $\overline{B}_j\neq \overline{V}_j$. Here we consider that users can only increase the number of VM requested, since when the number of VM decrease, the demand of users cannot be satisfied. Furthermore, we assume that the additional VMs have a utility 0. The main proof falls into several cases below.

\begin{itemize}
    \item \emph{Case 1 ($\overline{V}_j < P_j$ and $x_j^*=0$)}: We first consider the change in per-instance bid. If $b_j^i < v_j^i$, then $\sum_{i=1}^m k_j^i b_j^i< \sum_{i=1}^m k_j^i \rho_i$ (i.e., $\overline {B}_j < P_j$) and he also loses ($u_j=u_j^*=0$); else if $b_j^i > v_j^i$ but $\sum_{i=1}^m k_j^i b_j^i < \sum_{i=1}^m k_j^i \rho_i$, he still does not win; else $b_j^i > v_j^i$ and $\sum_{i=1}^m k_j^i b_j^i > \sum_{i=1}^m k_j^i \rho_i > \sum_{i=1}^m k_j^i v_j^i$, then he wins and his utility is $u_j=\sum_{i=1}^m k_j^i v_j^i-\sum_{i=1}^m k_j^i \rho_i < 0$. For increasing the number of requested VM with true valuation, if $\sum_{i=1}^m {k_j^i}' b_j^i < \sum_{i=1}^m {k_j^i}' \rho_i$ (i.e., $\overline {B}_j <  {P_j}'$) or $\sum_{i=1}^m {k_j^i}' b_j^i > \sum_{i=1}^m {k_j^i}' \rho_i$ but not satisfies the instance constraint, then $u_j=0$; else $\sum_{i=1}^m {k_j^i}' b_j^i > \sum_{i=1}^m {k_j^i}' \rho_i$ (i.e., $\overline {B_j} > {P_j}'$) and he wins, since when increasing the numbers of VMs, the utility of additional VMs is 0, and the total valuation does not increase (i.e., $\sum_{i=1}^m {k_j^i}' {v_j^i}' = \sum_{i=1}^m {k_j^i} v_j^i  $). So his utility is $u_j=\sum_{i=1}^m k_j^i v_j^i-\sum_{i=1}^m {k_j^i}' \rho_i < \sum_{i=1}^m k_j^i v_j^i-\sum_{i=1}^m {k_j^i} \rho_i <0$.

    \item \emph{Case 2~($\overline {V}_j \geq P_j$ and $x_j^*=0$)}: It means the requested VM instances does not satisfy the instance constraint. So given the random string $r$, no matter how the per-instance bid or the number of VM requested changes, he still won't win. Therefore, $u_j^*=u_j=0$.

    \item \emph{Case 3~($\overline {V}_j \geq P_j$  and $x_j^*=1$)}: If $b_j^i<v_j^i$ or $b_j^i>v_j^i$, and $\sum_{i=1}^m k_j^i b_j^i \geq \sum_{i=1}^m k_j^i \rho_i$, then his utility stays the same as the true valuation, $i.e., u_j=u_j^*$; else if $b_j^i<v_j^i$ and $\sum_{i=1}^m k_j^i b_j^i < \sum_{i=1}^m k_j^i \rho_i$, then he loses and his utility is $u_j=0<u_j^*$; else if the user $j$ increases the number of VM requested ($i.e., {k_j^i}'>k_j^i$), then $\sum_{i=1}^m {k_j^i}' b_j^i > \sum_{i=1}^m {k_j^i}' \rho_i$, he wins and his utility is $u_j=\sum_{i=1}^m k_j^i v_j^i - \sum_{i=1}^m {k_j^i}' \rho_i < \sum_{i=1}^m k_j^i v_j^i - \sum_{i=1}^m k_j^i \rho_i=u_j^*$; else the increasing VM instances does not meet the instance constraint or $\sum_{i=1}^m {k_j^i}' v_j^i < \sum_{i=1}^m {k_j^i}' \rho_i$  , then $u_j=0<u_j^*$.

\end{itemize}

Therefore, it can be proved that no user can improve his utility by bidding untruthfully, i.e., $u_j(V_j, \mathbf{B_{-j}})\geq u_j(B_j, \mathbf{B_{-j}})$.

\end{proof}

\begin{mytheorem}\label{the:gamma-truthfulness}
Given a random string $r$, DPCA is $\gamma$-truthful, where $\gamma$ =$\epsilon \cdot m \cdot q_{max} v_{max}$.
\end{mytheorem}
\begin{proof}
From Theorem~\ref{the:dp} and Lemma~\ref{the:truthfulness}, we can get $\exp(-\epsilon) Pr(\mathcal{M}(\mathbf{B'})=\rho) \le Pr(\mathcal{M}(\mathbf{B})=\rho)$ and $u_j(V_j, B_{-j}, r, \rho)\geq u_j(B_j, B_{-j}, r, \rho)$. Then the following formula holds,
\begin{displaymath}\small
        \begin{aligned}
       \qquad & E[u_j(V_j, \mathbf{B_{-j}}, r, \rho)]\\
        =\quad&\sum_{\rho \in {\Pi}^m} Pr(r) \times Pr[\mathcal{M}(V_j, \mathbf{B_{-j}})=\rho] \times u_j(V_j, \mathbf{B_{-j}}, r, \rho)\\
        \geq \quad &exp(-\epsilon)\sum_{\rho \in {\Pi}^m} Pr(r) Pr[\mathcal{M}(B_j, \mathbf{B_{-j}})=\rho]\cdot u_j(B_j, \mathbf{B_{-j}}, r, \rho)\\
        =\quad&exp(-\epsilon)E[u_j(B_j, \mathbf{B_{-j}}, r, \rho)]\\
        \geq \quad &(1-\epsilon)E[u_j(B_j, \mathbf{B_{-j}}, r, \rho)]\\
         =\quad&E[u_j(B_j, \mathbf{B_{-j}}, r, \rho)]-\epsilon E[u_j(B_j, \mathbf{B_{-j}}, r, \rho)]
        \end{aligned}
\end{displaymath}

For any user $j \in \mathbb{N}$, the ranges of unit price and true valuation are $[v_{min}, v_{max}]$ and the requested instances are limited in the interval of $[0,q_{max}]$, and the utility of user $j$ is $u_j=\sum_{i=1}^m k_j^i v_j^i- \sum_{i=1}^m k_j^i \rho_i$. Then
\begin{displaymath}
\begin{aligned}
E[u_j(B_j, \mathbf{B_{-j}}, r, \rho)]
\leq\quad&\max u_j
\leq\quad&m \cdot q_{max} v_{max}
\end{aligned}
\end{displaymath}

So,
$E[u_j(V_j, \mathbf{B_{-j}}, r, \rho)]\geq E[u_j(B_j, \mathbf{B_{-j}}, r, \rho)]-\epsilon m q_{max} v_{max}$

It satisfies the definition of approximate truthfulness. Therefore, we have proved Theorem \ref{the:gamma-truthfulness}.
\end{proof}

\begin{mylemma}\label{the:OPT}
Let $OPT^*=max_{\rho \in \Pi^m} REV(\mathbf{B}, r, \rho)$ denote the maximum revenue of DPCA. And $OPT$ is defined the optimal revenue of the cloud auction, i.e., $OPT=max_{\rho \in \Pi^m, r \in \mathbb{R}} REV(\mathbf{B}, r, \rho)$, where $\mathbb{R}$ is the set of all possible random strings. Then $\frac{K_{min}-q_{max}+1}{K_{max}} OPT \leq OPT^* \leq OPT$.
\end{mylemma}
\begin{proof}
Assuming that $OPT=REV(\mathbf{B}, r_{opt}, \rho_{opt})$, then the revenue of DPCA corresponding to $\rho_{opt}$ is $OPT'=REV(\mathbf{B}, r, \rho_{opt})$. Therefore, we get $OPT'\leq OPT^* \leq OPT$. Since $REV=\sum_{j \in \mathcal{W}} \sum_{i=1}^m \rho_i k_j^i$ and the total number of instances of each type VM requested by all winners will not exceed $K_{max}$, we have $OPT\leq K_{max} \sum_{i=1}^m \rho_{opt_i}$. As for $\rho_{opt}$ in DPCA, the number of instances of each type VM requested is at least $K_{min}-q_{max}+1$, then we can get $OPT' \geq (K_{min}-q_{max}+1)\sum_{i=1}^m \rho_{opt_i}$. So we can derive the relationship between $OPT$ and $OPT'$ expressed as $OPT'\geq \frac{K_{min}-q_{max}+1}{K_{max}} OPT$. Thus, $\frac{K_{min}-q_{max}+1}{K_{max}} OPT \leq OPT^* \leq OPT$.

\end{proof}

\begin{mytheorem}\label{the:revenue}
The expected revenue of the CSP $E[REV(\mathbf{B},r, \rho)]$ is at least $\frac{K_{min}-q_{max}+1}{K_{max}} OPT -\frac{6\Delta}{\epsilon} ln(e+\frac{\epsilon|{\Pi}^m|OPT}{2 \Delta})$.
\end{mytheorem}
\begin{proof}
Let $OPT^*$ = $max_{\rho \in {\Pi}^m}$ $REV(\mathbf{B}, r, \rho)$ denote the maximum revenue for the mechanism $\mathcal{M}$. For a small constant $t>0$, define the sets $R_{t}=\{\rho \in {\Pi}^m:REV(\mathbf{B}, r, \rho)>OPT^*-t\}$ and $\overline{R}_{2t}=\{\rho \in {\Pi}^m: REV(\mathbf{B}, r, \rho)\leq OPT^*-2t\}$. Then we have
\begin{displaymath}
\begin{aligned}
Pr(\mathcal{M}(\mathbf{B}) \in \overline{R}_{2t})&\leq \frac{Pr(\mathcal{M}(\mathbf{B}) \in \overline{R}_{2t})}{Pr(\mathcal{M}(\mathbf{B}) \in R_{t})}\\
&=\frac{\sum_{\rho \in \overline{R}_{2t}}\frac{{exp}(\frac{\epsilon REV(\mathbf{B}, r, \rho)}{2\Delta})}{\sum_{{\rho}'\in {\Pi}^m}{exp}(\frac{\epsilon REV(\mathbf{B}, r, {\rho}')}{2\Delta})}}{\sum_{{\rho} \in R_{t}}\frac{{exp}(\frac{\epsilon REV(\mathbf{B}, r, {\rho})}{2\Delta})}{\sum_{{\rho}'\in {\Pi}^m}{exp}(\frac{\epsilon REV(\mathbf{B}, r, {\rho}')}{2\Delta})}}\\
&=\frac{\sum_{\rho \in \overline{R}_{2t}}{{exp}(\frac{\epsilon REV(\mathbf{B}, r, \rho)}{2\Delta})}}{\sum_{\rho \in R_{t}}{exp}(\frac{\epsilon REV(\mathbf{B}, r, \rho)}{2\Delta})}\\
&\leq \frac{|\overline{R}_{2t}|{exp}(\frac{\epsilon(OPT^*-2t)}{2\Delta})}{|R_{t}|{exp}(\frac{\epsilon(OPT^*-t)}{2\Delta})}\\
&=\frac{|\overline{R}_{2t}|}{|R_{t}|}{exp}(\frac{-\epsilon t}{2\Delta})\\
&\leq |{\Pi}^m|{exp}(\frac{-\epsilon t}{2\Delta})\\
\end{aligned}
\end{displaymath}

Thus, $Pr(\mathcal{M}(\mathbf{B})\in R_{2t})\geq 1-|{\Pi}^m|{exp}(\frac{-\epsilon t}{2\Delta})$. And if $t$ satisfies the constraint that $t\geq \frac{2\Delta ln(\frac{|{\Pi}^m|OPT^*}{t})}{\epsilon}$, then $Pr(\mathcal{M}(\mathbf{B}) \in R_{2t})\geq 1-\frac{t}{OPT^*}$.

Therefore, for any $t\geq \frac{2\Delta ln(\frac{|{\Pi}^m|OPT^*}{t})}{\epsilon}$, we have
\begin{equation}\small\label{eq:expected-rev}
\begin{aligned}
E_{\rho \in {\Pi}^m}[REV(\mathbf{B}, r, \rho)]&\geq \sum_{\rho \in R_{2t}}REV(\mathbf{B}, r, \rho)Pr(\mathcal{M}(\mathbf{B})=\rho)\\
&\geq (OPT^*-2t)(1-\frac{t}{OPT^*})\\
&\geq OPT^*-3t
\end{aligned}
\end{equation}

Let $t=\frac{2\Delta ln(e+\frac{\epsilon|{\Pi}^m|OPT^*}{2\Delta})}{\epsilon}\geq \frac{2\Delta}{\epsilon}$, and we have
\begin{displaymath}\small
\begin{aligned}
\frac{2\Delta ln(\frac{|{\Pi}^m|OPT^*}{t})}{\epsilon}&\leq \frac{2\Delta ln(e+\frac{\epsilon|{\Pi}^m|OPT^*}{2\Delta})}{\epsilon}\\
&=t
\end{aligned}
\end{displaymath}

Then put $t=\frac{2\Delta ln(e+\frac{\epsilon|{\Pi}^m|OPT^*}{2\Delta})}{\epsilon}$ into Eq.~\eqref{eq:expected-rev} and combine Lemma~\ref{the:OPT}, the expected utility of mechanism $\mathcal{M}$ can be rewritten as
\begin{displaymath}
\begin{aligned}
\qquad & E_{\rho \in {\Pi}^m}[REV(\mathbf{B}, r, \rho)]\\
\geq \quad  & OPT^*-3t\\
\geq \quad & OPT^*-3\frac{2\Delta ln(e+\frac{\epsilon|{\Pi}^m|OPT^*}{2\Delta})}{\epsilon}\\
=\quad &OPT^*-\frac{6\Delta}{{\epsilon}} ln(e+\frac{\epsilon|{\Pi}^m|OPT^*}{2\Delta})\\
\geq \quad & \frac{K_{min}-q_{max}+1}{K_{max}} OPT -\frac{6\Delta}{{\epsilon}} ln(e+\frac{\epsilon|{\Pi}^m|OPT}{2 \Delta})
\end{aligned}
\end{displaymath}

Therefore, Theorem \ref{the:revenue} is proved that DPCA achieves approximate maximization of the CSP's revenue. 
\end{proof}

Note that since $K_i\gg q_{max}$, $\frac{K_{min}-q_{max}+1}{K_{max}} OPT\approx \frac{K_{min}}{K_{max}} OPT$. And when $K_{min}\approx K_{max}$, Lemma~\ref{the:OPT} is expressed as $OPT\leq OPT^* \leq OPT$. It then holds that:
$$E_{\rho \in {\Pi}^m}[REV(\mathbf{B}, r, \rho)] \geq OPT-\frac{6\Delta}{{\epsilon}} ln(e+\frac{\epsilon|{\Pi}^m|OPT}{2\Delta}).$$

We next analyze the computational complexity of Algorithm~\ref{protocol:DPCA}. The calculation of the revenue includes $|\Pi^m|$ iterations. In each iteration, the computational complexity of winner candidate selection is $\mathcal{O}(nm)$. And the complexity of sorting candidates is $\mathcal{O}(n\log n)$ in general. In the instance allocation phase, its computation cost is $\mathcal{O}(nm)$. And the complexity of computing the revenue is $\mathcal{O}(nm)$. Therefore, the computational complexity of Algorithm~\ref{protocol:DPCA} is $\mathcal{O}(\max\{nm, n\log n\}\cdot |\Pi^m| )$.

\section{Improvement of DPCA}\label{sec:improve}
As proved previously, our baseline proposition DPCA demonstrate nice theoretical properties. However, its time complexity is exponential to $\Pi$. In this section, we improve our basic mechanism DPCA in terms of time efficiency and auction benefit.

\subsection{DPCA-S}\label{sec:DPCA-S}
In DPCA-S (S for single), we repeatedly adopt an instance of exponential mechanism to select the unit price of one type of VM until $m$ unit prices are determined. In this way, we avoid exhaustively searching all the combinations in $\Pi^m$.

\subsubsection{Design Details}
\noindent

In the design of DPCA-S, the domain of each unit price $\rho_l$ of $VM_l$ is also defined as $\Pi=[v_{min},v_{max}]$. The main difference is that the clearing price vector is determined by selecting the unit prices one by one, instead of selecting them as a combination. The procedure is also described in following two phases.

\textbf{(1) Price Vector Selection.} This phase selects the unit prices one by one with instances of exponential mechanism, to reduce the computation complexity caused by DPCA. The main challenge is how to design the utility function for each instance of exponential mechanism. Our observation is that when selecting the $l$th ($1 \le l \le m-1$) unit price, only the local information of the first $l$ unit prices can be seen, and thus the corresponding utility function should only depend on this information. Therefore, when selecting the $l$th ($1 \le l \le m-1$) unit price, we only compare users' bids and prices of the first $l$ types of VMs to estimate a partial revenue, and use it as the utility value. After this, when selecting the $m$th unit price, we can make use of the global information of unit prices, and apply the same method as the previous section to compute the revenue, using it as the utility value. This phase can be depicted in three steps as follows.

\textit{(a) Partial Revenue Computation.} Given the unit price $\rho_l \in \Pi$ and assuming that the unit prices of the first $l-1$ VM types have been selected as $\rho_1, \rho_2,..., \rho_{l-1}$, respectively, DPCA-S calculates the partial bid of user $j$ $(j \in \mathbb{N})$ with the bid information of the first $l$ VM types using Eq.~\eqref{eq:partial-bid}.
\begin{equation}\label{eq:partial-bid}
\overline{B}_j^l=\sum_{i=1}^l k_j^i b_j^i
\end{equation}
The corresponding partial price can be computed as follows.
\begin{equation}\label{eq:partial-price}
P_j^l=\sum_{i=1}^l k_j^i \rho_i
\end{equation}
Then we can find the set $\mathcal{W}_l=\{j|\overline{B}_j^l \geq P_j^l|j \in \mathbb{N}\}$ of partial winner candidates. 

Given $\rho_l$ and for the case of $1\leq l\leq m-1$, the partial revenue for the first $l$ VM types can be computed with Eq.~\eqref{eq:partial-revenue}.
\begin{equation}\label{eq:partial-revenue}
REV_{l}(\mathbf{B},\rho_l)=\sum_{j\in \mathcal{W}_l} \sum_{i=1}^l \rho_i k_j^i
\end{equation}

\textit{(b) Total Revenue Computation.} When computing the total revenue, we go into a similar process as described in Section~\ref{sec:dpca}-B. Specially, given the first $m-1$ unit prices $\rho_1, \rho_2,..., \rho_{m-1}$, we can varies $\rho_m$ and get $|\Pi|$ price vectors. The handling of the revenue computation goes similarly, except that we have only $|\Pi|$ price vectors instead of $|\Pi|^m$ ones, and we do not go into details. Let $\mathcal{W}$ denote the set of winners, given a price vector $\rho$ and the random string $r$, the total revenue can be computed as
\begin{equation}\label{eq:total-revenue}
REV_m(\mathbf{B}, \mathbf{K}, r, \rho_m)=\sum_{j\in \mathcal{W}} \sum_{i=1}^m \rho_i k_j^i
\end{equation}
Note that we do not write other unit prices in Eqs.~\eqref{eq:partial-revenue} and \eqref{eq:total-revenue} for simplicity. We will also not explicitly show $\mathbf{K}$ and $r$ in the following probability distribution calculations.


\textit{(c) Unit Price Selection.} According to the exponential mechanism, DPCA-S sets the probability of the unit price $\rho_l$ to be exponentially proportional to its corresponding revenue.

\begin{equation}\label{eq:probability_l}
Pr(\mathcal{M}_l(\mathbf{B}) = \rho_l)=\frac{{\exp}(\frac{\epsilon_l REV_l(\mathbf{B}, \rho_l)}{2\Delta_l})}{\sum_{\rho_l' \in \Pi }{\exp}(\frac{\epsilon_l REV_l(\mathbf{B}, \rho_l')}{2\Delta_l})}
\end{equation}
where $\Delta_l=\sum_{i=1}^l q_{max} v_{max} = l \cdot q_{max} v_{max}$, $\epsilon_l=\frac{\epsilon}{m}$ is the privacy budget, and $1 \le l \le m$.

After this, each unit price $\rho_l$ $(1 \le l \le m)$ is randomly selected based on its probability distribution. The detailed process is described in Algorithm~\ref{protocol:DPCA-S}.

\textbf{(3) Winner Computation.} This phase is exactly the same as that of DPCA.

\renewcommand{\algorithmicrequire}{\textbf{Input:}}  
\renewcommand{\algorithmicensure}{\textbf{Output:}}  

\begin{algorithm}[htb]
\caption{Single-Unit-Price Selection} \label{protocol:DPCA-S}
\begin{algorithmic}[1]

\REQUIRE  $\mathbf{B}$, $\mathbf{K}$, $\mathbb{N}$, $r$ and privacy budget $\epsilon$.

\ENSURE Final clearing price $\rho$ and Winners $\mathcal{W}$


\STATE Initialize $x_1=...=x_n=0$, $\mathcal{W}_1, ...,\mathcal{W}_m, \mathcal{W}\leftarrow \varnothing$
\STATE Define $\rho_l \in \Pi=[v_{min},v_{max}]$, $\epsilon_l=\frac{\epsilon}{m}$

\FOR{$l \leftarrow 1$ to $m$}

\FOR{$\rho_l \in \Pi$}

\STATE $W_l\leftarrow \{j|\overline{B_j^l} \geq P_j^l, j \in \mathbb{N}\}$

\IF{$l\leq m-1$}

\STATE $REV_l(\mathbf{B},\rho_l)=\sum_{j \in \mathcal{W}_l} \sum_{i=1}^l \rho_i k_j^i$

\ELSE


\STATE  Rank the $h$ candidates in $\mathcal{W}_l$ based on $r$

\FOR{$j \in \mathcal{W}_l$}

\IF{$\forall i \in \{1,...,m\}, \sum_{t=1}^{j-1}k_t^i x_t+k_j^i \leq K_i$}

\STATE $x_j=1, \mathcal{W} \leftarrow \mathcal{W} \bigcup j$

\ENDIF
\ENDFOR

\STATE $REV_l(\mathbf{B}, \mathbf{K}, r, \rho_l)=\sum_{j\in \mathcal{W}} \sum_{i=1}^l \rho_i k_j^i$

\ENDIF
\ENDFOR

\FOR{$\rho_l \in \Pi$}

\STATE $Pr(\mathcal{M}_l(\mathbf{B}) = \rho_l)=\frac{{\exp}(\frac{\epsilon_l REV_l(\mathbf{B}, \rho_l)}{2\Delta_l})}{\sum_{\rho_l' \in \Pi }{\exp}(\frac{\epsilon_l REV_l(\mathbf{B}, \rho_l')}{2\Delta_l})}$

\ENDFOR


\STATE $\rho_l \leftarrow \mathcal{M}_l(\mathbf{B})$

\ENDFOR

%
%
%
%
%
%
%
%
%
%
%
%
%
%
%


%

\RETURN $\rho=(\rho_1, \rho_2, ..., \rho_m)$

\end{algorithmic}
\end{algorithm}


\subsubsection{Analysis}
We now establish the economic properties of DPCA-S.

\begin{mytheorem}\label{the:dp_s}
DPCA-S achieves $\epsilon$-differential privacy.
\end{mytheorem}

\begin{proof}
Let $\mathcal{M}_l$ denote the mechanism to select the unit price $\rho_l$. For each mechanism $\mathcal{M}_l$, it is simply an application of the exponential mechanism, so it can achieve $\epsilon_l$-differential privacy. Then, according to the composition lemma (Lemma~\ref{the:composition}), DPCA-S achieves $\epsilon$-differential privacy, where $\epsilon=\epsilon_1+\epsilon_2+...+\epsilon_m$.
\end{proof}

\begin{mytheorem}\label{the:dp-st}
DPCA-S achieves $\gamma$-truthful, where $\gamma=\epsilon \cdot m \cdot q_{max} v_{max}$.
\end{mytheorem}

\begin{proof}
According to Theorem~\ref{the:dp_s}, we can get $\exp(-\epsilon) Pr(\mathcal{M_S}(\mathbf{B'})=\rho) \le Pr(\mathcal{M_S}(\mathbf{B})=\rho)$ where $\mathcal{M_S}$ represents the mechanism DPCA-S. Then the proof of truthfulness is the same as Theorem~\ref{the:gamma-truthfulness}.
\end{proof}

It seems hard to give a theoretical bound for the revenue of DPCA-S. Instead, we will provide an experimental comparison between the revenues of DPCA-S and DPCA in Section~\ref{sec:experiment}.

Finally, we analyze the computational complexity of the Algorithm~\ref{protocol:DPCA-S}. The computational complexity of choosing the prices of the first $m-1$ types VM is $\mathcal{O}(nm \cdot (m-1)\cdot \Pi)$. And the complexity of the last price's choice is $\mathcal{O}(\max\{nm,n\log n\} \cdot \Pi)$. Thus, the computational complexity of the Algorithm~\ref{protocol:DPCA-S} is $\mathcal{O}(\max\{nm,n\log n\} \cdot m \cdot \Pi)$. Therefore, the complexity is greatly reduced from the exponential level to linear level.

\subsection{DPCA-M}\label{sec:DPCA-M}
In DPCA-S, we reduce the price space by selecting the unit prices one by one. However, this benefit comes at the price of reducing auction revenue because of large noise introduced in every unit price selection. A natural question is whether we can get a better auction result with an acceptable time efficiency. Following this line of thinking, we design a differentially private mechanism for combinatorial cloud auctions with multiple unit price selection, called DPCA-M.

The main idea is that we regard $t$ unit prices as a group, each time applying an instance of exponential mechanism to select a group of prices. Thus, we just have to use $\lceil m/t \rceil$ exponential mechanism instances to confirm these $m$ unit prices. To realize the tradeoff between time and benefits, let $t$ satisfy the following inequality, $i.e., 1 < t < m$. Actually, it can be interpreted as a generalization of DPCA and DPCA-S in the sense that DPCA and DPCA-S are the extreme cases when $t=m$ and $t=1$, respectively.

\subsubsection{Design Details}
In order to select a group of $t$ prices, we can define each group set $\widetilde{\rho}_l=(\rho_{(l-1)t+1}, ... ,\rho_{l t})$ in the domain of $\Pi^t=[v_{min},v_{max}]^t$. For the last price group $\widetilde{\rho}_{\lceil m/t\rceil}={(\rho_{(\lceil m/t\rceil-1)t+1},...,\rho_m)}$, the number of types of VM may be less than $t$, so its set is defined as $\Pi^{(m-t (\lceil m/t\rceil-1))}$. The design can also be described in two phases as follows.

\textbf{(1) Price Vector Selection.} This phase is basically similar to the unit price selection of DPCA-S. The significant difference is that we view $t$ unit prices as a group, and randomly pick them out together at a time. Therefore, in the process of selecting $\widetilde{\rho}_l$, we should calculate the partial bids and partial prices of the first $l \cdot t$ $(1 \le l < \lceil m/t\rceil)$ VM types, or in the end the total bids and total prices of all $m$ VM types. The phase can be divided into two steps.

\textit{(a) Price Group Selection with Partial Revenue.} For $1 \le l < \lceil m/t\rceil$, we define the utility of the exponential mechanism as
\begin{equation}\label{eq:bid2}
REV_{l}(\mathbf{B},\widetilde{\rho}_l)=\sum_{j\in \mathcal{W}_l} \sum_{i=1}^{lt} \rho_i k_j^i
\end{equation}
Thus, a mechanism $\mathcal{M}_l$ chooses each price group $\widetilde{\rho}_l$ with the probability as
\begin{equation}\label{eq:group-probability-l}
Pr(\mathcal{M}_l(\mathbf{B}) = \widetilde{\rho}_l)=\frac{{\exp}(\frac{\epsilon_l REV(\mathbf{B}, \widetilde{\rho}_l)}{2\Delta_l})}{\sum_{\widetilde{\rho}_l' \in \Pi }{\exp}(\frac{\epsilon_l REV(\mathbf{B}, \widetilde{\rho}_l')}{2\Delta_l})}
\end{equation}
where $\Delta_l=l \cdot t \cdot q_{max} v_{max}$ and $\epsilon_l=\frac{\epsilon}{\lceil m/t\rceil}$ is the privacy budget.

\textit{(b) Price Group Selection with Total Revenue.} For $l = \lceil m/t\rceil$, this step first determines a winner candidate set $\mathcal{W}_l$ as previous step, and then uses a similar procedure as that of DPCA to allocation VM with instance constraints, obtaining the winner set $\mathcal{W}$, and finally computes the total revenue by Eq.~\eqref{eq:group-total-bid} as the utility of the exponential mechanism.
\begin{equation}\label{eq:group-total-bid}
REV_{l}(\mathbf{B}, \mathbf{K}, r, \widetilde{\rho}_l)=\sum_{j\in \mathcal{W}} \sum_{i=1}^{m} \rho_i k_j^i
\end{equation}
Then, the last unit price group is selected by Eq.~\eqref{eq:group-probability-l} with the sensitivity $\Delta_{\lceil m/t \rceil}=m \cdot q_{max} v_{max}$.

\textbf{(2) Winner Computation.} This step is the same as that of DPCA.

\subsubsection{Analysis}
Because DPCA-M can be viewed as a generalization of DPCA-S, it is easy to see that DPCA-M achieves differential privacy and approximate truthfulness. These proofs are similar as those of DPCA-S and we omit them. For the analysis of revenues, we again leverage the experimental method.

Additionally, the computational complexity of DPCA-M is $\mathcal{O}(\max\{nm,n\log n\} \cdot \lceil  m /t \rceil\cdot \Pi^t)$, which is a trade-off between DPCA and DPCA-S.

%

\section{Performance Evaluation}\label{sec:experiment}

In this section, we fully implement our basic mechanism DPCA, improved mechanisms DPCA-S and DPCA-M, and do extensive experiments to evaluate their performances.

\subsection{Experimental Setting}

In the experimental setting, we consider the following two scenarios due to the range of per-instance bids and the numbers of VM types.
\begin{itemize}
\item \emph{Small Scale Scenario.} We let per-instance bids of users be generated randomly from the interval $[0, 10]$, and use small numbers of VM types (e.g. 2-6). This scenario is mainly used for running DPCA, since its computational complexity is $\Omega(\Pi^m)$, and when the size of $\Pi$ and $m$ are bigger, the execution time of DPCA may become unbearable. This also explains the necessity of DPCA-S and DPCA-M.
\item \emph{Practical Scenario.} We let per-instance bids be generated randomly from the interval $[0, 100]$, and apply more practical numbers of VM types (e.g. 20). This scenario is mainly used for evaluating the performances of DPCA-S and DPCA-M.
\end{itemize}

Other parameters are set as follows. The numbers of VM instances $K_i (i \in \{1,2,...,m\})$ and those requested by users $k_j^i$ are uniformly distributed over $[K_{min},K_{max}]$ and $[0,10]$, respectively. By default, the total privacy budget $\epsilon$ is 1 and $T$ represents the combination of $T$ unit prices of DPCA-M. Note that $\epsilon=1$ is the total privacy budget, the privacy budget for each exponential mechanism of DPCA-S is $\epsilon/m$ and that of DPCA-M is $\epsilon/\lceil m/T\rceil$. The experimental results are the averaged over 100 trials. For the performance evaluation, we adopt the following metrics:

\begin{itemize}
\item \emph{Revenue}: The sum of the price paid by all winning users.
\item \emph{User Satisfaction}: The ratio of the number of winning users to the total number of all users.
\item \emph{Running Time}: The time spent to execute an auction.

\end{itemize}

\subsection{Experimental Results}

When evaluating the performance of the proposed mechanisms, we compare them with a truthful cloud auction mechanism without privacy guarantee (denoted by "Basic") in \cite{Wang2012When}.

\begin{figure}[htbp]
\begin{center}
\includegraphics[width=0.5\textwidth]{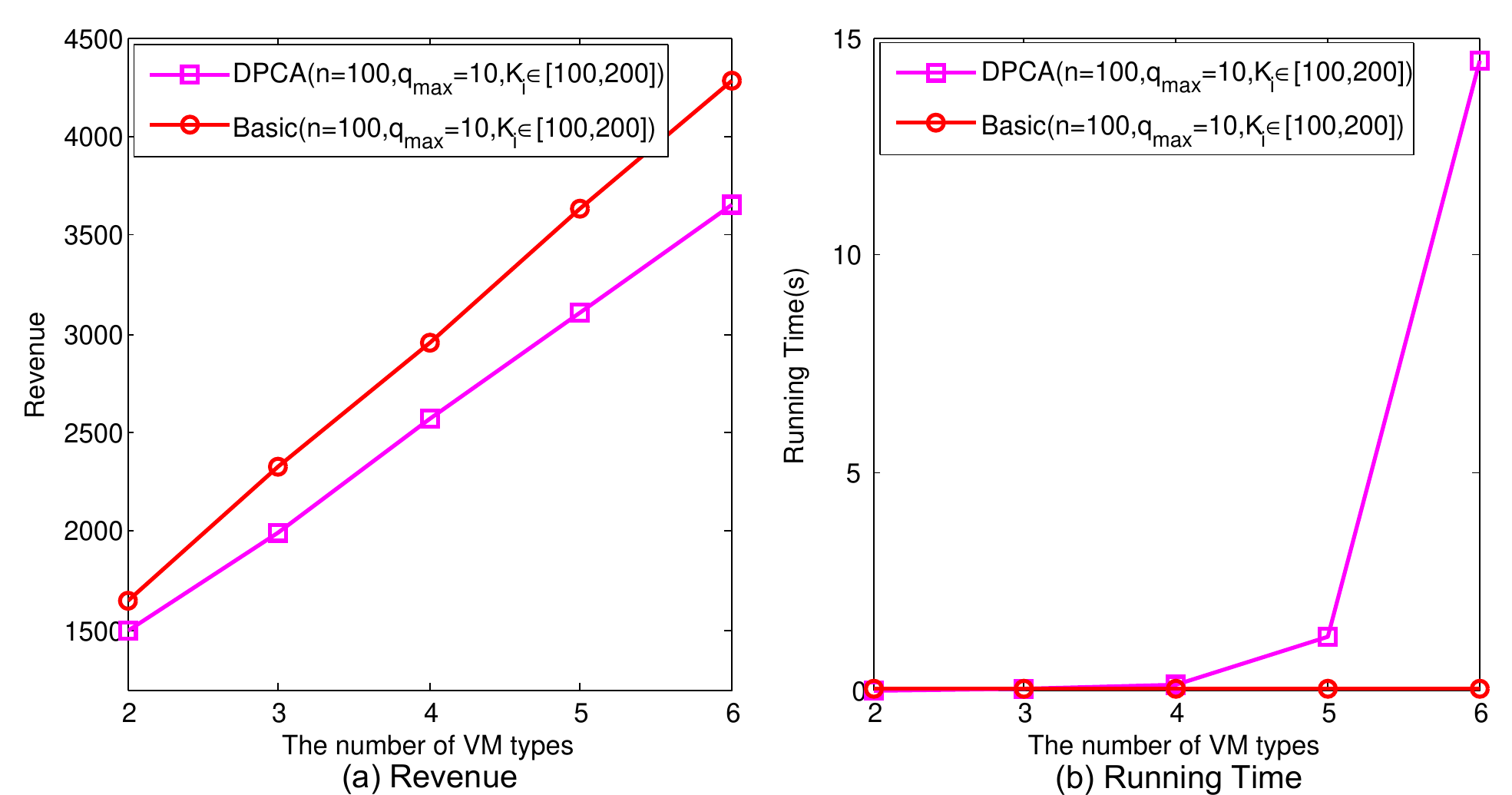}
\caption{Comparison between between DPCA and Basic}\label{fig:types}
\end{center}
\end{figure}

\begin{figure*}[htbp]
\begin{center}
\includegraphics[width=0.8\linewidth]{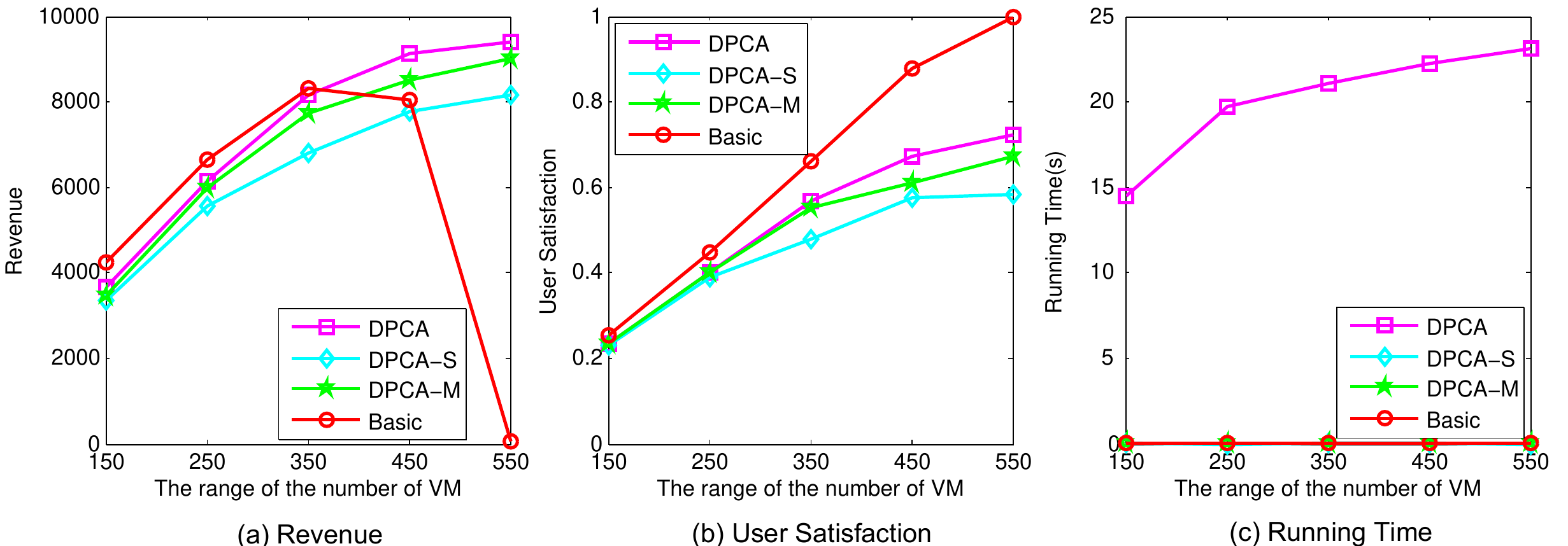}
\caption{Revenue, user satisfaction and running time as the range of the number of VM grows}\label{fig:kmin}
\end{center}
\end{figure*}

\begin{figure*}[ht]
\centering
\begin{minipage}[b]{0.485\linewidth}
\includegraphics[width=1\textwidth]{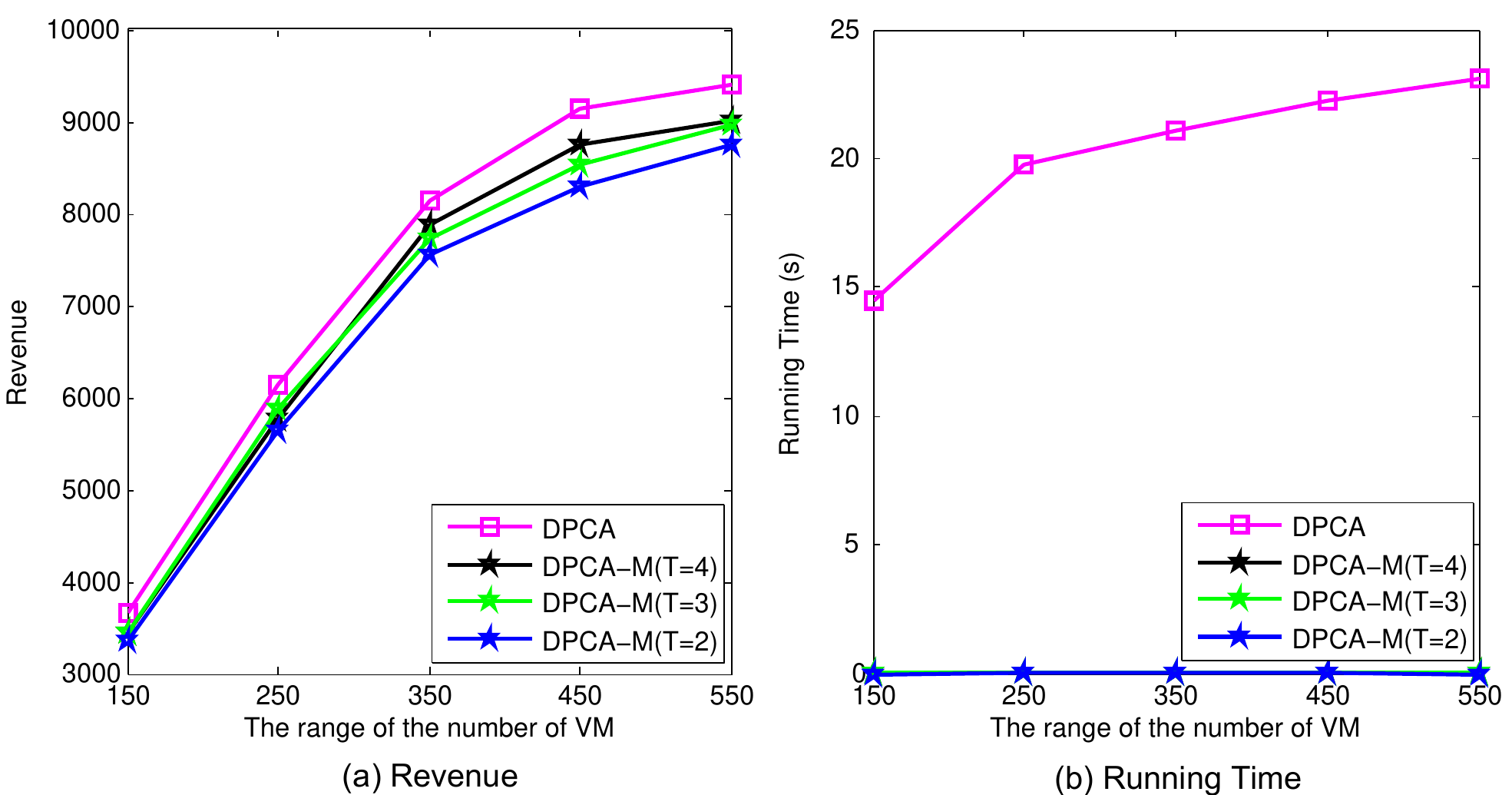}
\caption{Comparison between between DPCA and DPCA-M}\label{fig:group}
\label{fig:minipage1}
\end{minipage}
\quad
\begin{minipage}[b]{0.485\linewidth}
\includegraphics[width=1\textwidth]{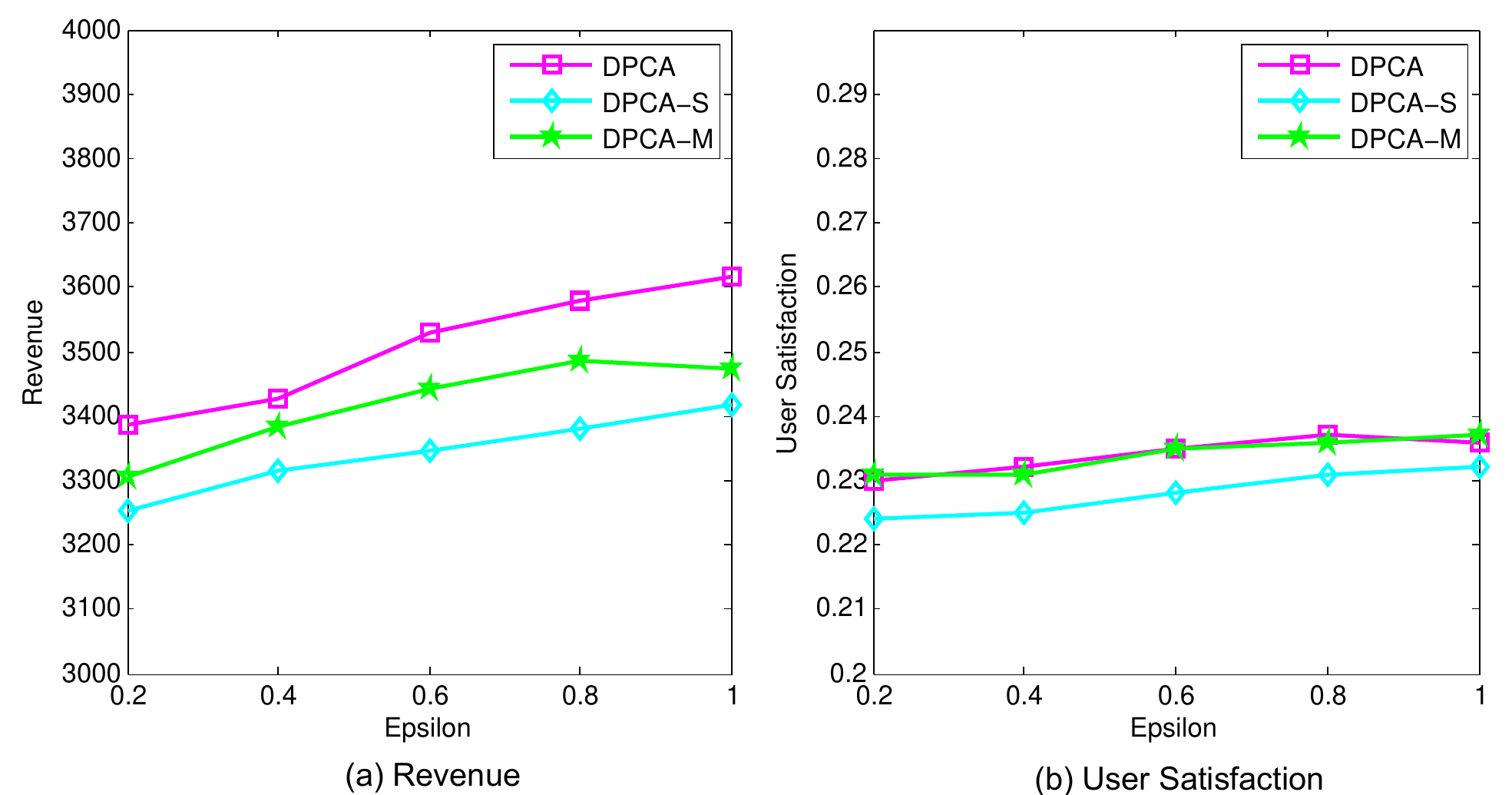}
\caption{Revenue and user satisfaction as $\epsilon$ grows}\label{fig:epsilon}
\label{fig:minipage2}
\end{minipage}
\end{figure*}

\begin{figure*}[htbp]
\begin{center}
\includegraphics[width=0.78\linewidth]{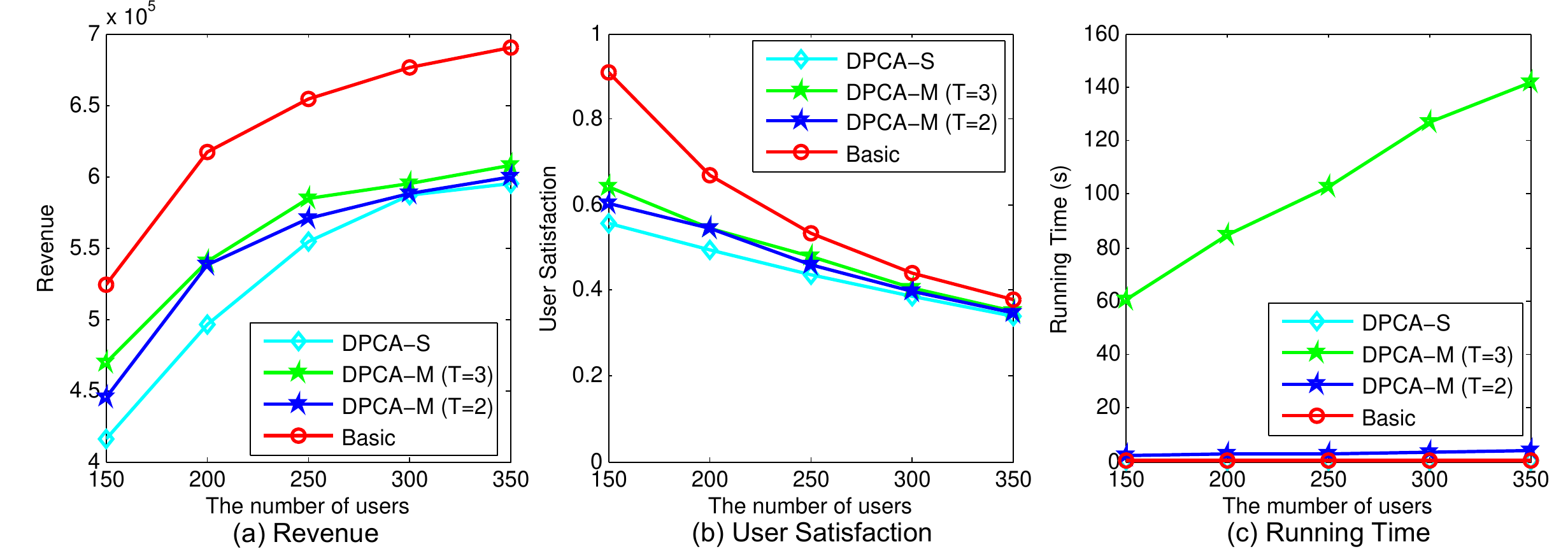}
\caption{The performance of DPCA-S, DPCA-M and Basic at $m=20$ and the per-instance bid interval $[0,100]$}\label{fig:large}
\end{center}
\end{figure*}

\textit{(1) Performance comparisons in small scale scenario.}

In Fig.~\ref{fig:types}, we plot the revenues and running times of DPCA and Basic as the number of VM types varies from $2$ to $6$, with $n=100$ and $K_i \in [100,200]$. From Fig.~\ref{fig:types}(a), we can see that the revenues of both mechanisms increase as $m$ increases, and Basic's revenue is superior to that of DPCA, since DPCA sacrifices some of its revenue to achieve differential privacy. In Fig.~\ref{fig:types}(b), the running time of Basic is basically maintained at around $10$ millisecond, while that of DPCA increases very fast as $m$ grows. This is because that, with the increase of $m$, the space for the unit price combination becomes exponentially large, which leads to that the running time of DPCA also gets large rapidly. In a word, the experimental results demonstrate that DPCA provides privacy preservation with a mild revenue cost, but it is only suitable for auctions with small numbers of VM types, otherwise the running time would grow exponentially fast.


Fig.~\ref{fig:kmin} illustrates the comparisons of the revenues, user satisfactions and running times for DPCA, DPCA-S, DPCA-M and Basic, as the range $[K_{min},K_{max}]$ varies from $[100,200]$ to $[500,600]$, and $m=6$, $n=100$. Note that we use the midpoint (i.e., $\frac{K_{min}+K_{max}}{2}$) of a range to represent the range. From Fig.~\ref{fig:kmin}(a), we make the following observations. (a) the revenue of Basic increases as the numbers of VMs increase before the range $[300,400]$, after which the revenue drops steeply to nearly zero. (b) The revenues of DPCA, DPCA-S, DPCA-M all grow with the increase of numbers of VMs and eventually exceeds that of Basic. (c) DPCA performs the best, DPCA-M with $T=3$ takes the second place, and DPCA-S is the last one.

The reason for observation (a) is that the payment scheme of Basic is based on $\emph{critical users}$. As the numbers of VM increase, the VMs becomes oversupplied and more and more users win, and it is hard to find a $\emph{critical user}$ for each winner. Eventually, when all users win, the revenue becomes 0. For observation (b), the cause is that our mechanisms DPCA, DPCA-S, DPCA-M are priced according to the exponential mechanism. After selecting the final clearing unit price vector, more winners will generate higher revenue. The reason underlying observation (c) may be as follows. DPCA employs only an instance of exponential mechanism, DPCA-M employs more, and DPCA-S employs the most. The more instances of exponential mechanism employed, the more noise introduced, resulting in poorer revenues. Furthermore, DPCA uses global bid information for pricing, while DPCA-M and DPCA-S use partial bid information. The partial information may deteriorates the revenues.

Expectedly, in Fig.~\ref{fig:kmin}(b), user satisfaction of all mechanisms grows with the number of VM. Basic's increases to $100\%$ but others gradually approaches an equilibrium. The reason is as for observation (a). And Fig.~\ref{fig:kmin}(c) depicts the running time of all mechanisms. DPCA takes much longer than the other mechanisms and the running time increases as the number of VM groups while the time of the other mechanisms remain roughly constant between 10 and 20 milliseconds. The running times shown are as expected.

The experimental results in Fig.~\ref{fig:kmin} indicate that on the basis of privacy guarantee, DPCA can obtain the best revenues with the longest running time, DPCA-S produce the worst revenues but with the least running time, while DPCA-M achieves in-between revenues and in-between running times. Meanwhile, when the number of VM provided is much higher than the number requested by users, our proposed mechanisms generates higher revenue than the mechanism without privacy protection.

Fig.~\ref{fig:group} depicts the revenues and running times for DPCA and DPCA-M under different $T$ as $K_i$ changes. And we fix $m=6$ and $n=100$, respectively. As shown in Fig.~\ref{fig:group}(a), the revenues of all mechanisms increases as $K_i$ increases. For DPCA-M, the larger $T$ is, the higher revenue is. And DPCA which can be view as $T=6$ generates the highest revenue. From Fig.~\ref{fig:group}(b), we can see that the running time of DPCA-M under different $T$ remains basically consistent and is less than $30$ milliseconds, while that of DPCA is more than $10$ seconds, which increase linearly with the increase of $K_i$. The above observations show that for DPCA-M, the more unit prices in the group, the higher the revenue.

Fig.~\ref{fig:epsilon} traces the revenues and user satisfactions of our mechanisms as the privacy budget $\epsilon$ varies from $0.2$ to $1$, when $m=6$, $n=100$, $K_i \in [100,200]$ and DPCA-M under $T=3$. It can be seen that with the increase of privacy budget $\epsilon$, both revenue and user satisfaction of three mechanisms show increasing trends. 
It demonstrates that the higher the privacy budget, the better the benefit of the auction.

\textit{(2) Performance comparisons in practical scenario.}

Fig.~\ref{fig:large} further illustrates the performance of DPCA-S and DPCA-M as the number of users increases from 150 to 350 in practical scenarios where $m=20$, $K_i \in [300,400]$ and the range of the per-instance bids is $[0,100]$. Under this configuration, we cannot run DPCA because the running time is quite unbearable. Expectedly, in Fig.~\ref{fig:large}(a), the revenues of all mechanisms increase as the number of users raises. Also the curve of Basic outperforms that of the other three mechanisms. For DPCA-S and DPCA-M under different $T$, the revenue of DPCA-M under $T=2$ is superior to that of DPCA-S, but is inferior to that of DPCA-M under $T=3$. From Fig.~\ref{fig:large}(b), we can observe that the user satisfaction of all mechanisms is gradually decreasing and eventually getting close when the supply is getting low. In Fig.~\ref{fig:large}(c), it is obvious that the running time of DPCA-M under $T=3$ is much higher than that of other mechanisms. And the time of DPCA-S is minimal and no more than $200$ milliseconds, while DPCA-M under $T=2$ spends more time than Basic, but does not exceed $5$ seconds. All the time is practically acceptable. The above experimental results show that both DPCA-S and DPCA-M are applicable to the practical scenarios and can generate good revenue. Moreover, for DPCA-M, the more unit prices in the group, the longer time consumed. And DPCA-M with an appropriate $T$ will provide a good trade-off between running times and auction revenues.

\section{Conclusion}\label{sec:conclusion}
In this paper, we have proposed a differentially private mechanism for combinatorial cloud auctions, called DPCA. To achieve differential privacy, appriximate truthfulness and high revenue, we randomly select the final clearing unit price vector based on the specific probability distribution. Through theoretical analysis, we prove the properties in privacy, truthfulness and revenue. Considering the time and benefits, we further develop DPCA-S and DPCA-M. We conduct simulations to evaluate their performance. The experimental results demonstrate that DPCA can generate better revenue but is only appropriate for small-scale cloud auctions, while DCPA-M is more suitable in practical scenarios.

\bibliography{sample-bibliography}{}

\begin{thebibliography}{10}
\providecommand{\url}[1]{#1}
\csname url@samestyle\endcsname
\providecommand{\newblock}{\relax}
\providecommand{\bibinfo}[2]{#2}
\providecommand{\BIBentrySTDinterwordspacing}{\spaceskip=0pt\relax}
\providecommand{\BIBentryALTinterwordstretchfactor}{4}
\providecommand{\BIBentryALTinterwordspacing}{\spaceskip=\fontdimen2\font plus
\BIBentryALTinterwordstretchfactor\fontdimen3\font minus
  \fontdimen4\font\relax}
\providecommand{\BIBforeignlanguage}[2]{{%
\expandafter\ifx\csname l@#1\endcsname\relax
\typeout{** WARNING: IEEEtran.bst: No hyphenation pattern has been}%
\typeout{** loaded for the language `#1'. Using the pattern for}%
\typeout{** the default language instead.}%
\else
\language=\csname l@#1\endcsname
\fi
#2}}
\providecommand{\BIBdecl}{\relax}
\BIBdecl

\bibitem{Amazon.org}
Amazon, ``Amazon ec2 spot instances,'' \url{http://aws.amazon.com/ec2/spot/}.

\bibitem{Wang2012When}
Q.~Wang, K.~Ren, and X.~Meng, ``When cloud meets ebay: Towards effective
  pricing for cloud computing,'' in \emph{Infocom, IEEE}, 2012.

\bibitem{Zaman2010Combinatorial}
S.~Zaman and D.~Grosu, ``Combinatorial auction-based allocation of virtual
  machine instances in clouds,'' in \emph{IEEE Second International Conference
  on Cloud Computing Technology and Science}, 2010.

\bibitem{Shi2014RSMOA}
W.~Shi, C.~Wu, and Z.~Li, ``Rsmoa: A revenue and social welfare maximizing
  online auction for dynamic cloud resource provisioning,'' in \emph{Quality of
  Service}, 2014.

\bibitem{Chen2016On}
Z.~Chen, C.~Lin, L.~Huang, and Z.~Hong, ``On privacy-preserving cloud
  auction,'' in \emph{Reliable Distributed Systems}, 2016.

\bibitem{Dwork2006Differential}
C.~Dwork, ``Differential privacy,'' in \emph{International Colloquium on
  Automata, Languages, and Programming}, 2006, pp. 1--12.

\bibitem{mcsherry2007mechanism}
F.~McSherry and K.~Talwar, ``Mechanism design via differential privacy,'' in
  \emph{Foundations of Computer Science, 2007. FOCS'07. 48th Annual IEEE
  Symposium on}.\hskip 1em plus 0.5em minus 0.4em\relax IEEE, 2007, pp.
  94--103.

\bibitem{zhu2014differentially}
R.~Zhu, Z.~Li, F.~Wu, K.~Shin, and G.~Chen, ``Differentially private spectrum
  auction with approximate revenue maximization,'' in \emph{Proceedings of the
  15th ACM international symposium on mobile ad hoc networking and
  computing}.\hskip 1em plus 0.5em minus 0.4em\relax ACM, 2014, pp. 185--194.

\bibitem{Jin2016Enabling}
H.~Jin, S.~Lu, B.~Ding, K.~Nahrstedt, and N.~Borisov, ``Enabling
  privacy-preserving incentives for mobile crowd sensing systems,'' in
  \emph{IEEE International Conference on Distributed Computing Systems}, 2016.

\bibitem{Jian2017BidGuard}
L.~Jian, D.~Yang, L.~Ming, X.~Jia, and G.~Xue, ``Bidguard: A framework for
  privacy-preserving crowdsensing incentive mechanisms,'' in
  \emph{Communications and Network Security}, 2017.

\bibitem{xu2017pads}
J.~Xu, B.~Palanisamy, Y.~Tang, and S.~M. Kumar, ``Pads: Privacy-preserving
  auction design for allocating dynamically priced cloud resources,'' in
  \emph{Collaboration and Internet Computing (CIC), 2017 IEEE 3rd International
  Conference on}.\hskip 1em plus 0.5em minus 0.4em\relax IEEE, 2017, pp.
  87--96.

\bibitem{Zaman2014A}
S.~Zaman and D.~Grosu, ``A combinatorial auction-based mechanism for dynamic vm
  provisioning and allocation in clouds,'' \emph{IEEE Transactions on Cloud
  Computing}, vol.~1, no.~2, pp. 129--141, 2014.

\bibitem{Mashayekhy2015A}
L.~Mashayekhy, M.~M. Nejad, and D.~Grosu, ``A ptas mechanism for provisioning
  and allocation of heterogeneous cloud resources,'' \emph{IEEE Transactions on
  Parallel and Distributed Systems}, vol.~26, no.~9, pp. 2386--2399, 2015.

\bibitem{Wang2013Revenue}
W.~Wang, B.~Liang, and B.~Li, ``Revenue maximization with dynamic auctions in
  iaas cloud markets,'' in \emph{IEEE/ACM International Symposium on Quality of
  Service}, 2013.

\bibitem{Zhang2014Dynamic}
L.~Zhang, Z.~Li, and C.~Wu, ``Dynamic resource provisioning in cloud computing:
  A randomized auction approach,'' in \emph{Infocom, IEEE}, 2014.

\bibitem{du2019learning}
B.~Du, C.~Wu, and Z.~Huang, ``Learning resource allocation and pricing for
  cloud profit maximization,'' in \emph{The Thirty-Third AAAI Conference on
  Artificial Intelligence (AAAI-19)}, 2019.

\bibitem{Wei2013Dominant}
W.~Wei, B.~Li, and B.~Liang, ``Dominant resource fairness in cloud computing
  systems with heterogeneous servers,'' 2013.

\bibitem{Hong2013A}
Z.~Hong, L.~Bo, H.~Jiang, and F.~Liu, ``A framework for truthful online
  auctions in cloud computing with heterogeneous user demands,'' in
  \emph{Infocom, IEEE}, 2013.

\bibitem{Mashayekhy2016An}
L.~Mashayekhy, M.~M. Nejad, D.~Grosu, and A.~Vasilakosu, ``An online mechanism
  for resource allocation and pricing in clouds,'' \emph{IEEE Transactions on
  Computers}, vol.~65, no.~4, pp. 1172--1184, 2016.

\bibitem{jiao2019auction}
Y.~Jiao, P.~Wang, D.~Niyato, and K.~Suankaewmanee, ``Auction mechanisms in
  cloud/fog computing resource allocation for public blockchain networks,''
  \emph{IEEE Transactions on Parallel and Distributed Systems}, 2019.

\bibitem{cheng2019towards}
K.~Cheng, Y.~Slien, Y.~Zhang, X.~Zhu, L.~Wang, and H.~Zhong, ``Towards
  efficient privacy-preserving auction mechanism for two-sided cloud markets,''
  in \emph{ICC 2019-2019 IEEE International Conference on Communications
  (ICC)}.\hskip 1em plus 0.5em minus 0.4em\relax IEEE, 2019, pp. 1--6.

\bibitem{zhu2015differentially}
R.~Zhu and K.~G. Shin, ``Differentially private and strategy-proof spectrum
  auction with approximate revenue maximization,'' in \emph{Computer
  Communications (INFOCOM), 2015 IEEE Conference on}.\hskip 1em plus 0.5em
  minus 0.4em\relax IEEE, 2015, pp. 918--926.

\bibitem{8676063}
Z.~{Chen}, T.~{Ni}, H.~{Zhong}, S.~{Zhang}, and J.~{Cui}, ``Differentially
  private double spectrum auction with approximate social welfare
  maximization,'' \emph{IEEE Transactions on Information Forensics and
  Security}, vol.~14, no.~11, pp. 2805--2818, Nov 2019.

\bibitem{jin2018privacy}
X.~Jin and Y.~Zhang, ``Privacy-preserving crowdsourced spectrum sensing,''
  \emph{IEEE/ACM Transactions on Networking (TON)}, vol.~26, no.~3, pp.
  1236--1249, 2018.

\bibitem{jin2018dpda}
W.~Jin, M.~Li, L.~Guoy, and L.~Yang, ``Dpda: A differentially private double
  auction scheme for mobile crowd sensing,'' in \emph{2018 IEEE Conference on
  Communications and Network Security (CNS)}.\hskip 1em plus 0.5em minus
  0.4em\relax IEEE, 2018, pp. 1--9.

\bibitem{gao2019dpdt}
G.~Gao, M.~Xiao, J.~Wu, S.~Zhang, L.~Huang, and G.~Xiao, ``Dpdt: A
  differentially private crowd-sensed data trading mechanism,'' \emph{IEEE
  Internet of Things Journal}, 2019.

\bibitem{Nisan2007Algorithmic}
N.~Nisan, T.~Roughgarden, E.~Tardos, and V.~V. Vazirani, \emph{Algorithmic game
  theory}.\hskip 1em plus 0.5em minus 0.4em\relax Cambridge university press,
  2007.

\bibitem{Gupta2010Differentially}
A.~Gupta, K.~Ligett, F.~McSherry, A.~Roth, and K.~Talwar, ``Differentially
  private combinatorial optimization,'' in \emph{Proceedings of the
  twenty-first annual ACM-SIAM symposium on Discrete Algorithms}.\hskip 1em
  plus 0.5em minus 0.4em\relax SIAM, 2010, pp. 1106--1125.

\bibitem{dwork2008differential}
C.~Dwork, ``Differential privacy: A survey of results,'' in \emph{International
  Conference on Theory and Applications of Models of Computation}.\hskip 1em
  plus 0.5em minus 0.4em\relax Springer, 2008, pp. 1--19.

\bibitem{dwork2014algorithmic}
C.~Dwork, A.~Roth \emph{et~al.}, ``The algorithmic foundations of differential
  privacy,'' \emph{Foundations and Trends{\textregistered} in Theoretical
  Computer Science}, vol.~9, no. 3--4, pp. 211--407, 2014.

\end{thebibliography}
\bibliographystyle{IEEEtran}

\end{document}